\documentclass[10pt,conference]{IEEEtran}
\IEEEoverridecommandlockouts

\def\BibTeX{{\rm B\kern-.05em{\sc i\kern-.025em b}\kern-.08emT\kern-.1667em\lower.7ex\hbox{E}\kern-.125emX}}

\usepackage{epsfig,booktabs}
\usepackage{verbatim}
\usepackage{amssymb}
\usepackage{amsmath}
\usepackage{color}
\usepackage{graphicx}
\usepackage{subfig}
\usepackage{wrapfig}
\usepackage{enumerate}
\usepackage{hyperref}

\newtheorem{prop}{Proposition}

\newtheorem{theorem}{Theorem}
\newtheorem{definition}{Definition}
\newtheorem{lemma}{Lemma}

\newcommand{\sq}{\hbox{\rlap{$\sqcap$}$\sqcup$}}
\newcommand{\qed}{\hspace*{\fill}\sq}
\newenvironment{proof}{\noindent {\bf Proof.}\ }{\par\vskip 4mm\par}

\begin{document}

\title{\LARGE \bf Partitionable Asynchronous Cryptocurrency Blockchain}

\author{
Kendric Hood, Joseph Oglio, Mikhail Nesterenko, and Gokarna Sharma\\
\textit{Department of Computer Science, Kent State University,
Kent, OH 44242, USA} \\
khood5@kent.edu, joglio@kent.edu, mikhail@cs.kent.edu, sharma@cs.kent.edu}

\maketitle 
\thispagestyle{plain}
\pagestyle{plain}

\begin{abstract}
We consider operation of blockchain-based cryptocurrency in case of partitioning. We define the Partitionable Blockchain Consensus Problem. The problem may have an interesting solution if the partitions proceed independently by splitting accounts. 
We prove that this problem is not solvable in the asynchronous system. The peers in the two partitions may not agree on the last jointly mined block or, alternatively, on the starting point of independent concurrent computation.
We introduce a family of detectors that enable a solution. We establish the relationship between detectors. We present the algorithm that solves the Partitionable Blockchain Consensus Problem using our detectors. We extend our solution to multiple splits, message loss and to partition merging. We simulate and evaluate the performance of detectors, discuss the implementation of the detectors and future work.
\end{abstract}


\sloppy

\section{Introduction} Peer-to-peer networks are an attractive way to organize distributed computing.
Blockchain is a technology for building a shared, immutable, distributed ledger. Blockchain is typically maintained by a peer-to-peer network.
A prominent blockchain application is cryptocurrency such as Bitcoin~\cite{nakamoto} or Ethereum~\cite{ethereum}. In a blockchain cryptocurrency, the ledger records financial transactions.
The transactions are appended to the ledger block by block. The application is decentralized and peers have to agree on each block. Such consensus is the foundation of blockchain algorithms. 

Classic robust distributed consensus algorithms~\cite{pbft,Byzantine} use cooperative message exchange between peers to arrive at a joint decision. However, such algorithms require participants to know the identity of all peers in the network. This is not often feasible for modern high-turnover peer-to-peer networks.

Bitcoin uses Nakamoto consensus~\cite{nakamoto} where the participants compete to add blocks to the ledger. This algorithm does not require complete network knowledge. Due to the simplicity and robustness of the algorithm, Nakamoto consensus became widely-used in cryptocurrency design.

For their proper operation, blockchain consensus algorithms assume that the network remains connected at all times. The network may not confirm transactions while parts of the network are unable to communicate. Alternatively, a single primary partition makes progress while the others are not utilized. This is not accidental as the problem of concurrently using partitions is difficult to handle: the partitioned peers may approve transactions that conflict and thus violate the integrity of the combined blockchain. Ultimately, the problem is stated in the classic CAP Theorem~\cite{gilbert2002brewer} claiming that it is impossible co-satisfy consistency, availability and partition-tolerance in a distributed system. 

It is assumed that partitioning interruptions are infrequent or insignificant for network operation. However, this may not be the case. While long network-wide splits may be rare, brief separations are common. This can happen, for example, if groups of peers are connected via a small set of channels. If these channels are congested, the groups are effectively cut off from each other and the network is temporarily partitioned. 

As blockchain becomes a larger component of the financial market, the pressure for the system's availability will increase. Many financial applications, such as high frequency trading \cite{kirilenko2017flash}, are sensitive to even a slight delay. A system delay that lasts a few milliseconds may cost its users substantial time and money. Thus, considering the blockchain-based cryptocurrency that is available through partitioning is required for this technology to realize its potential. 

\ \\
\textbf{Our contribution.} In this paper, we formally state the Partitionable Blockchain Consensus Problem. What enables the solution is the possibility of splitting the accounts and processing transactions in the partitions independently without violating the integrity of the complete blockchain.

We use the asynchronous system model for the study of partitionable consensus. The model does not place assumptions on the 
peers' 
relative computation power or message propagation delay and thus have near-universal applicability. 
Consensus is impossible in the asynchronous system even if a single peer crashes~\cite{FLP}. Intuitively, peers are not able to distinguish a crashed process from a very slow one. This impossibility is circumvented with crash failure detectors~\cite{chandra1996weakest,chandra1996unreliable} that provide minimum synchrony to allow a solution. 
We pattern our investigation on this classic approach. We show the impossibility of a partitionable consensus solution in the asynchronous system and then introduce partitioning detectors to enable it.  We present an algorithm that solves the Partitionable Blockchain Consensus Problem using the detectors. 
To simplify the presentation, we first consider a single split with no message loss and no subsequent merging. We then extend our solution to accommodate message loss, multiple splits and temporary splits and merging. We implement our algorithm and evaluate its performance.

\ \\
\textbf{Related work.}
There is 
a number of recent publications dealing with the implementation or modification of classic consensus~\cite{sbft,hyperledger,algorand,honeybadger}.
There are plenty of recent studies presenting blockchain design based on Nakamoto consensus~\cite{bentov2016snow,bitcoinng,ouroboros,nakamoto,pass2017fruitchains}. 
There are papers that combine classic and Nakamoto consensus~\cite{abraham2016solida,hybrid}. Recent research on Nakamoto-based blockchain often focuses on improving its speed and scalability~\cite{decker2016bitcoin,elastico,pass2018thunderella,rapidchain}. One promising blockchain acceleration technique is to concurrently build a DAG of blocks~\cite{popov2016tangle,ghost}. The state-of-the-art on the blockchain consensus algorithms can be found in this recent survey \cite{xiao2019}. 

Relatively few publications focus on partitionable blockchains. There are some studies where the partitionable classic consensus either addressed directly~\cite{friedman1997fast} or 
using failure detectors ~\cite{Aguilera1999,dolevPODC97}.
Partitionable consensus has similarities with the group membership problem, which deals with presenting a consistent membership set to the processes despite process and link failures \cite{babaouglu1997group,chandra1996impossibility,chockler2001group}

Extended virtual synchrony (EVS) \cite{EVS} is a technique that supports continued operation in all partitions. That is, during network partitioning and re-connection, it maintains a consistent relationship between the delivery of messages and configuration changes. However, static membership is assumed and hence this is not easily extendable to work in blockchain systems with dynamic membership. 

Tran {\it et al.} \cite{Tran2019SwarmDAGAP} consider an algorithm that implements partitionable blockchain consensus in the context of swarm robotics. In swarm robotics, the robot swarms may experience network partitions due to navigational and communication challenges or in order to perform certain tasks efficiently. Their solution extends EVS and hence is not suitable for partitionable blockchain under dynamic membership as we consider in this paper. Recently, Guo {\it et al.} \cite{GuoPS19} 
observed that synchronous classic consensus protocols cannot even tolerate a short-term jitter that takes a node offline or makes it leave the system for a very short time. 
Then, they provided a solution that makes those protocols tolerant to such jitter while keeping the same consistency and liveness properties.
Karlsson {\it et al.} propose a partitionable blockchain for Internet-of-things devices~\cite{karlsson2018vegvisir}.

\section{Notation}
\label{section:notation}
\noindent
\textbf{Communication model.} A \emph{peer} is a single process. All peers operate correctly and do not have faults. A \emph{partition} is a collection of peers that can communicate. This communication is done through message passing. A broadcast sends a message to every peer in its partition. Communication channels have infinite capacity and are FIFO. The channels are reliable unless the network is split. A \emph{(network) split} separates one partition into two. 
Peers are split into two arbitrary non-empty sets. A message sent before the network split is always delivered; a message sent after the split is delivered only to the recipients that are in the same partition as the sender. 

Each peer contains a set of variables and commands. A \emph{network state} is an assignment of a value to each variable from its domain. An \emph{action} is an execution of a command.  An action transitions the network from one state to another.
A \emph{computation} is a sequence of states resulting from actions.  Actions in a computation are atomic and do not overlap. We consider the network split to be a particular kind of action. A computation is either infinite or ends in a state where no action may be executed. A \emph{computation segment} is a portion of a computation that starts and ends in a state. A computation segment from the initial state to a particular state is a computation \emph{prefix}. A, possibly infinite, computation segment following a particular state is a \emph{suffix}.

We assume fair action execution and fair message receipt. Specifically, in any computation, any action is eventually either executed or disabled; any sent message is eventually received. We assume there is at most one split per computation. That is, the network starts as one partition and it may split into two. Observe that this means that fairness does not apply to a network split action: the split may not happen. Since we are considering the purely asynchronous system model, there may be no re-connections. If a split occurs, the two partitions exist for the rest of the computation. We relax the single split and no re-connection assumption further in the paper. 

\ \\
\noindent\textbf{Causality.} An action $a_1$ \emph{causally precedes} action $a_2$ if either (i) $a_1$ and
$a_2$ are different actions of the same process and $a_1$ occurs before
$a_2$, (ii) $a_1$ contains a send and $a_2$ contains a receipt of
the same message (iii) $a_1$ is a receipt of a message sent outside of the partition and $a_2$ is the network split creating this partition. 
While the first two of the above cases are fairly conventional~\cite{lamport1976ordering}, the last case may require clarification. Indeed, once the split occurs, no messages sent outside the partition are received. Hence, all such receipts causally precede the split.

The casual precedence relation is transitive.  Two actions that are not causally related are \emph{concurrent}. Since a split does not affect communication inside a partition, intra-partition communication is concurrent with the split. Also, since there is no communication between partitions, any two actions in separate partitions are concurrent. 

If the following sequence $\langle\cdots,a_i,a_{i+1}\cdots\rangle$ is
a computation of some algorithm, and the actions $a_i$ and $a_{i+1}$ are concurrent, then $\langle\cdots,a_{i+1},a_{i},\cdots\rangle$ is also a computation of this algorithm.  That is, swapping consequent concurrent actions in a computation of an algorithm, produces another computation of the algorithm. 

\ \\
\noindent
\textbf{Accounts and transactions.} An \emph{account} is a means of storing funds. Each account has a unique identifier. A \emph{transaction} is a transfer of funds from the source to the target account. For simplicity, we assume that there is a single source and a single target account. Each transaction has a unique identifier as well.

A \emph{client} is an entity that submits transactions to the network via broadcast. There may be multiple clients. The transaction identifiers for each client are monotonically increasing. A client submits each transaction to a single partition. If transactions are submitted to two partitions, they are considered separate transactions. 

\ \\
\noindent
\textbf{Blockchain.} Peers mine transactions. Such a mined transaction is a \emph{block}. That is, to simplify the presentation, we assume a single transaction per block. A mined block cannot be altered. Besides a transaction, each block contains an identifier of another block. Thus, a block is linked to another block. A \emph{blockchain} is a collection of such linked blocks. A \emph{genesis} is the first block in the blockchain. The genesis is unique. There are no cycles in the blockchain. That is, the blockchain is a tree. A \emph{branch} of a tree is a chain of blocks from the genesis to one of the leaves of the tree. See, for example, a branch from the genesis to block $1$ in Figure~\ref{figPARTImage}.

The \emph{main chain} is the longest branch in a blockchain. Ties are broken deterministically. A \emph{permanent branch} is infinite. We assume that there is at most one permanent branch per partition.

Each peer operates as follows. If it receives a transaction, it stores it. The peer attempts to mine one of the pending transactions by linking it to the tail of its main chain. If it succeeds, the block is immediately broadcast. The peer continues while there are pending transactions. A peer may quit trying to mine a transaction and switch to mining another one. For example, if a new block arrives, a peer may switch to mining on top of it. We make the following assumption:  if a peer receives infinitely many new transactions, then the peer either receives infinitely many mined blocks or mines infinitely many blocks itself.

\emph{Global blockchain (tree)} is the collection of all blocks mined by all peers. A \emph{fork} is the case of multiple blocks linking to the same block. A fork happens when several peers succeed in concurrently mining blocks. See Figure~\ref{figPARTImage} for an example. Since the genesis block is unique, it may never be in a fork. If there is a network split, a \emph{seed} is the last block mined on any branch before split. That is, at the time of a split, the last block on every branch is a seed. 

Consider a block $b$ on the blockchain tree and a branch that leads from the genesis to $b$.
A \emph{balance} for any particular account $a$ with respect to $b$ is the sum of all funds that are transferred into $a$ by the transactions of the blocks in this branch minus the funds that are transferred out of $a$.

A transaction is \emph{valid} if its application leaves the source account with a non-negative balance. The transaction is \emph{invalid} otherwise. The transactions may differ depending on the particular branch of the tree.  Therefore, a transaction may be valid in one branch and invalid in another. If a peer mines a transaction, it is valid relative to its main chain. That is, peers mine only valid transactions. 

A transaction is \emph{confirmed} if it is in the permanent branch. It is \emph{rejected} if it is not in the permanent branch. A transaction is \emph{resolved} if it is either confirmed or rejected.  A transaction is \emph{permanently valid} if it is valid indefinitely or until it is resolved. We assume that in each partition, at least one client submits infinitely many permanently valid transactions.
 
\noindent
 \textbf{Account splits.}
In case of a network split, the account balances may also be split. That is, the peers consider the amount of funds available on a particular account to be a fraction of the original amount. Accounts are split in the seed: the last pre-split block in the blockchain. The blockchain may have multiple branches and, therefore, multiple seeds. Thus, accounts may potentially be split in different seeds.
See Figure~\ref{figPARTImage} for example.

To eliminate a trivial case, we assume that at least some funds are distributed between partitions. That is, we exclude the case where a partition is left with zero funds for all accounts. Otherwise, we place no restrictions on the way that accounts are split between partitions so long as the total on each account balance in both partitions post-split is equal to the pre-split account balance in the seed block. For example, suppose there is an account $a$ that has a balance of $100$. Then a network split occurs. Account $a$ may be split into $a_1$ and $a_2$. Accounts $a_1$ and $a_2$ cannot be in the same partition. Account $a$ is split $70/30$, thus the balance of $a_1$ is $70$ and $a_2$ is $30$. If the accounts are split unevenly, each peer must know to which partition it belongs. 
If $a$ is halved, i.e. split $50/50$, then the peers do not need to identify which partition they are in.
%
%
Observe that a split affects the validity of a transaction. If the submitted transaction is not valid after split, it is not mined. 

A \emph{branch merge} is an arbitrary interleaving of transactions of two branches such that the order of transactions of each branch is preserved. Two branches are \emph{mergeable} if all transactions mined before the split are resolved uniformly and any branch merge retains the validity of all transactions of the two branches.

\begin{prop} \label{propMerge}
Branches from different partitions are mergeable if they are split on the same seed. 
\end{prop}

\ \\
\noindent\textbf{Detectors.}
A \emph{detector} is a mechanism that provides information to the algorithm that it may not able to determine otherwise. 
Specifically, a detector is an algorithm whose actions may include information available outside the model. The \emph{pure asynchronous} system has no detectors.  The output variables of a detector provide information for other algorithms to use. The actions of the detector and the algorithm that uses it, interleave in a fair manner. 

Detector \emph{A} is \emph{weaker} than detector \emph{B}, if there exists an algorithm such that it accepts every computation of 
\emph{A} and produces a computation of \emph{B}. Detector \emph{A} is \emph{equivalent} to \emph{B}, denoted $B \equiv A$, if both \emph{A} is weaker than \emph{B} and \emph{B} is weaker than \emph{A}. Detector \emph{A} is \emph{strictly weaker} than \emph{B}, denoted $B \succ A$, if \emph{A} is weaker than \emph{B} but not equivalent to \emph{B}.

\section{The Partitionable Blockchain Consensus Problem}
\label{section:problem}

\begin{definition} {\em The Partitionable Blockchain Consensus problem is the intersection of the following three properties: \\
{\bf confirmation validity:} no invalid transaction is confirmed;\\
{\bf branch compatibility:} permanent branches are mergeable;\\
{\bf progress:} every permanently valid transaction is eventually confirmed.\\
}
\end{definition}
The first two properties are safety while the progress property is liveness~\cite{liveness}.

%
%

\section{Impossibility}
\label{section:impossibility}

We show that it is not possible to achieve partitionable blockchain consensus in 
the pure asynchronous system. 
The intuition for our argument is as follows. In the pure asynchronous system, peers may not directly know whether the split has occurred. The peers may only infer this from message communication.
Recall that our model guarantees reliable pre-split message delivery while after split this guarantee is only within the sender's partition. Thus, a sender may not be certain whether all the peers in the network received its message. We exploit this uncertainty to demonstrate lack of solution to the Partitionable Blockchain Consensus Problem.

A transaction is \emph{split-invalidated} if it is valid unless a split occurs. For example, if an account has a balance of $10$, and this account is split $50/50$, then a transaction that spends $6$ is split-invalidated. 

To exclude a trivial solution which rejects all split-invalidated transactions, we introduce the following definition. An algorithm is \emph{regular} if there exits a computation of this algorithm where a split-invalidated transaction mined before the split is confirmed. 
An algorithm is \emph{strictly regular} if it confirms all split-invalidated transactions mined before the split.

\begin{lemma}\label{lemSplitInvalidated}
A regular algorithm that solves the Partitionable Blockchain Consensus Problem may not confirm a split-invalidated transaction.
\end{lemma}

Lemma~\ref{lemSplitInvalidated} states that a regular solution to the Partitionable Blockchain Consensus Problem may not resolve a split-invalidated transaction. However, this means that there is no regular solution to this problem at all. Hence the below theorem.   

\begin{theorem}\label{thrmNoPure}
There does not exist a regular solution to the Partitionable Blockchain Consensus Problem in the pure asynchronous system.
\end{theorem}
\noindent

\begin{proof}
Assume there is a regular algorithm \emph{ALG} that solves the Partitionable Blockchain Consensus Problem. Since \emph{ALG} is regular, there is a computation $c_x$  of \emph{ALG} that contains a split-invalidated confirmed transaction $t$ that is mined before the split.  Let $a(t)$ and $a(s)$ be actions of $c_x$ such that: $a(t)$ is the action that mines $t$, $a(s)$ is the action that splits the network into two partitions $p_1$ and $p_2$.

Let us examine $c_x \equiv \langle \mathit{pfx}, a(t), seg, a(s), \cdots \rangle$, where $\mathit{pfx}$ is a prefix preceding the mining of transaction $a(t)$, $seg$ --- a segment of the computation separating $a(t)$ and the split action $a(s)$.
Let us further examine the actions of $seg$. We focus on the actions sending messages between peers in $p_1$ and $p_2$. Consider
$seg \equiv \langle seg_1, a(tr), seg_2 \rangle $, where $seg_1$ is a segment that  contains arbitrary actions, $a(tr)$ is the last action in $seg$ that transmits a message between peers in partitions $p_1$ and $p_2$. That is, segment $seg_2$ does not contain any message sent from a peer in one of the future partitions to a peer in the other, i.e. it contains intra-partition communication only.

Consider another computation of $c_y$ with the following prefix $\langle \mathit{pfx}, a(t), seg_1, a(s) \cdots\rangle $. That is, $c_y$ shares the prefix with $c_x$ up to the last sending of a message between the two partitions. The split happens right before the possible message transmission. Computation $c_y$ may either reject $t$ or confirm it. If $c_y$ confirms $t$ just like $c_x$, we continue the process of moving the split action $a(s)$ and shortening the segment $seg$. We stop when we find a computation that rejects $t$ or when we exhaust the segment. Hence, there could be two possible cases.

\emph{Case 1}. There exist two computations $c_u$ and $c_v$ such that $c_u$ rejects $t$ while $c_v$ confirms it and the composition of the two computations is
as follows:
\begin{align*}
&c_u \equiv \langle \mathit{pfx}, a(t), seg_3, a(s), \mathit{sfx_u}\rangle\ \  \textrm{and} \\
&c_v \equiv \langle \mathit{pfx}, a(t), seg_3, a(tr), seg_4, a(s), \mathit{sfx_v} \rangle
\end{align*}
where $seg_3$ is a segment that contains arbitrary actions, $a(tr)$ is an action sending messages between partitions $p_1$ and $p_2$, and $seg_4$ is a segment with only intra-partition communication. That is, the two computations differ in the way they resolve $t$ and share a prefix up to the split action except for a single transmission between partitions and intra-partition communication $a(tr)$. To put another way, this single transmission determines whether transaction $t$ is accepted or rejected.

Assume, without loss of generality, that $a(tr)$ is an action of a peer in $p_1$ that sends messages from partition $p_1$ to partition $p_2$.  Now, the receipt of this transmission effectively determines whether peers in $p_2$ accept or reject $t$. We exploit this dependency in the construction of the following computation $c_{uv}$. 

In this construction, $seg|_p$ denotes the actions from segment $seg$ by peers in partition $p$. The two rows indicate the actions of peers in partitions $p_1$ and $p_2$ respectively. 
\begin{equation*}
c_{uv} \equiv \bigg\langle \mathit{pfx}, a(t), seg_3, a(s), 
\begin{array}{l}
a(tr), seg_4|_{p1}, \mathit{sfx_v}|_{p1} \\
\mathit{sfx_u}|_{p2}
\end{array}
\bigg\rangle 
\end{equation*}

That is, the prefix of $c_{uv}$ is $\langle \mathit{pfx}, a(t), seg_3, a(s) \rangle$. 
In the partition $p_1$, the actions are as in $c_v$ following the split except the split is moved ahead of $a(tr)$ and actions of $seg_4$.  Let us discuss this move. Action $a(tr)$ contains a message transmission from a peer in $p_1$ to peers $p_2$. Otherwise, $a(tr), seg_4|p_1$ contain only intra-partition communication and are thus concurrent with $a(s)$. These actions can be swapped with the split without affecting causality.
To put another way, the peers in $p_1$, may not determine whether the message in $a(tr)$ is received by peers in $p_2$. For the partition $p_2$, we add to $c_{uv}$ the actions as in $c_u$ following the split. We interleave the actions of the two partitions in $c_{uv}$ in an arbitrary but fair manner.

Let us examine $c_{uv}$. By construction, it is a computation of \emph{ALG}. However, in partition $p_1$, the actions of peers are as in $c_v$. This means that they confirm $t$. But, in partition $p_2$, the peers behave as in $c_u$. This means that they reject $t$. This however, means that the two partitions do not agree on the resolution of a transaction that is mined pre-split. That is, this computation violates the Branch Compatibility Property of the Partitionable Blockchain Consensus Problem. Let us now consider the second case.

\emph{Case 2.}
In every computation where the split happens after mining $t$, $t$ is confirmed regardless of communication between partitions. 

Let us then consider the computation $c_z \equiv \langle 
\mathit{pfx}, a(t), a(s), \mathit{sfx_z} \rangle$. That is, in this computation the split happens right after $a(t)$ is mined. Such a computation confirms $t$ as well. Assume, without loss of generality, that $t$ is mined in partition $p_1$.

Let us now examine a computation $c_w \equiv \langle \mathit{pfx}, a(s), a(t),\mathit{sfx_w}\rangle$. This computation shares a prefix with $c_z$ but the split occurs before the $t$ is mined. In this case, the $t$ is invalid. Therefore, \emph{ALG} has to reject $t$ in $c_w$.

We construct a computation $c_{zw}$ as follows. 
\begin{equation*}
c_{zw} \equiv \bigg\langle \mathit{pfx}, a(s), 
\begin{array}{l}
a(t), \mathit{sfx}_z|_{p1} \\
\mathit{sfx}_w|_{p2}
\end{array} \bigg\rangle 
\end{equation*}

The prefix of this computation is $\mathit{pfx}$. Then, in partition $p_1$, this prefix is followed by actions of $c_z$ in partition $p_1$ except the order of $a(t)$ and $a(s)$ is reversed. That is, the mining of the transaction follows the split. Since the split affects the communication between partitions, this changing of the order of the two transactions does not affect the actions in $p_1$. Since the split precedes the mining of the transaction, it is not received by the peers $p_2$. For the partition $p_2$, we add to $c_{zw}$ the actions of $c_w$ following the split. 
Note that $a(t)$ is mined in $p_1$, therefore, it only goes into the actions of partition one. 
The actions of the two partitions are interleaved in an arbitrary but fair manner. 

Let us consider the constructed computation $c_{zw}$. Since the actions in $p_1$ are as in $c_z$, the peers in this partition confirm $t$. However, the actions in $p_2$ are as in $c_w$. This means the peers in this partition reject $t$. That is, in $c_{zw}$, the peers in the two partitions disagree on the resolution of the pre-split transaction. That is, this computation violates the Branch Compatibility Property.

To summarize, any regular algorithm \emph{ALG} that confirms a split-invalidated transaction also has computations that violate the properties of the Partitionable Blockchain Consensus Problem. Hence, contrary to our initial assumption, a solution to the Partitionable Blockchain Consensus Problem may not confirm a split-invalidated transaction.
\qed\end{proof}

\section{Partitioning Detectors}
\label{section:detectors}

The partitionable blockchain consensus problem is impossible without detectors. The lack of solution is due to the impossibility of ascertaining whether the message reached all peers. Let us discuss detectors that may circumvent this and enable a solution. A propagation detector \emph{PROP} addresses this concern directly: for each peer and for each block, \emph{PROP} outputs whether this block is delivered to peers of the entire network or just for a single partition. 

Let us give a more precise specification for \emph{PROP}. Assume \emph{PROP} is running concurrently with a blockchain consensus algorithm \emph{ALG}. In \emph{PROP}, each peer contains an input variable $in$ and an output variable $out$. Initially, both variables contain $\bot$. Algorithm \emph{ALG} can write to \emph{in} and read from \emph{out}. Let $b$ be a block mined by \emph{ALG} in some computation. 
If there is a suffix of the computation of \emph{PROP} where $in=b$, then the computation of \emph{PROP} consists of a prefix where $out=\bot$ followed by a suffix where either (i) $out=\textbf{true}$ if the block in \emph{ALG} was received by all peers in the network or (ii) a suffix where $out=\textbf{false}$ if the block was delivered only to a partition. That is, \emph{PROP} correctly classifies block receipts, does not make mistakes or changes its decision. This definition is extended to an arbitrary number of blocks in a straightforward manner: there is an input and output variable for each block in each peer. 

Another detector \emph{AGE} outputs whether the block was mined before or after the split. The detector classifies the pre-split block as \emph{old}, and post-split block as \emph{new}. The formal definition is similar to that of \emph{PROP} above. 

Since the block is broadcast right after mining and message transmission is reliable, the only way for a message not to be delivered to all peers is if there is a split in the network. Hence the following lemma.

\begin{lemma} 
The propagation detector is equivalent to the age detector. That is: $PROP \equiv AGE$.
\end{lemma}

Detector eventual \emph{AGE}, denoted $\diamond AGE$, is similar to \emph{AGE}. Like \emph{AGE}, detector $\diamond AGE$ outputs whether the block was mined before or after the split. However, $\diamond AGE$ is not reliable: $\diamond AGE$ may make a finite number of mistakes. Specifically, given a block mined by algorithm \emph{ALG} running concurrently with $\diamond AGE$, if there is a suffix of computation of $\diamond AGE$ such that where $in = b$, then the computation of $\diamond AGE$ contains a suffix where $out=new$ if the block is mined before split or $out = old$ if the block is mined post-split. 

To put another way, for each peer and for each block, $\diamond AGE$ is guaranteed to eventually output the correct result. To distinguish $\diamond AGE$, we call \emph{AGE} the perfect age detector. 

\begin{lemma}\label{lem:ageWeaker}
The $\diamond AGE$ detector is strictly weaker than $AGE$. That is: $\diamond AGE \prec AGE$.
\end{lemma}

\begin{proof}
To prove strict weakness of $\diamond AGE$, we need to show that $AGE$ is not weaker than $\diamond AGE$. That is, there there does not exist an algorithm that takes any computation of $\diamond AGE$ and produces a computation of \emph{AGE}.

Assume the opposite. Let \emph{ALG} be such an algorithm. Let computation $c_1$ of $\diamond AGE$ decide the age of a single block $b$. The block is old.
As perfect \emph{AGE} does not make mistakes, for each peer it has to contain a prefix where $out=\bot$ followed by a suffix where $out = \text{old}$. Let $s_1$ be the state of $c_1$ where \emph{ALG} outputs the decision of \emph{AGE} for every peer of the network. This decision is based on the output of $\diamond AGE$.

We compose the computation $c_2$ as follows. It contains the same block $b$ and the same output of $\diamond AGE$ up to and including $s_1$. However, in this computation $b$ is new. That is, $\diamond AGE$ makes a mistake which it corrects later in the computation of $c_2$
However, since $c_1$ and $c_2$ share prefixes, \emph{ALG} outputs that $b$ is old. That is, \emph{ALG} makes a mistake. Perfect \emph{AGE} may not make mistakes. 
Therefore, \emph{ALG} may not be an implementation of \emph{AGE}. This means that our initial assumption is incorrect and the lemma follows.
\qed\end{proof}




Detector \emph{WAGE}, pronounced "weak-age", has output similar to \emph{AGE} and $\diamond AGE$. However, unlike these two detectors, \emph{WAGE} may make infinitely many mistakes subject to the following constraints. For each block for at least one peer per partition, the suffix of the computation of \emph{WAGE} contains only correct output; for all other peers, every suffix contains infinitely many states with correct output. 
To put another way, \emph{WAGE} ensures that at least one peer per partition eventually starts classifying blocks correctly and all other peers at least alternate their classifications without permanently settling on incorrect output.  However, $\diamond AGE$ may be implemented using only \emph{WAGE}. Figure~\ref{figAgeEquivalency} shows this implementation. Below lemma formalizes this statement.


\begin{lemma}\label{lemAgeWage}
 Eventual age detector is equivalent to weak age detector. That is: $ \diamond AGE \equiv \mathit{WAGE}$.
\end{lemma}

\begin{proof}
All computations of $\diamond AGE$ are already computations of \emph{WAGE}. To prove the equivalency, we need to show that $\diamond AGE$ can be implemented using $\mathit{WAGE}$. We discuss the implementation of the detector for a single block. For multiple blocks, the detector implementation runs concurrently. 

The implementation algorithm is shown in Figure~\ref{figAgeEquivalency}. It operates as follows. Each peer $p$, maintains the last known output of $\mathit{WAGE}$ for all peers. It is stored in array $ages$ indexed by peer identifier. Similarly, $p$ keeps track of the number of changes in the output of $\mathit{WAGE}$ for each peer. This is stored in array $\mathit{flips}$. If the output of $\mathit{WAGE}$ changes for its peer, it updates $ages[p]$, increments $\mathit{flips}[p]$ and broadcasts the update.
For implemented output of $\diamond AGE$, each peer outputs the value of $\mathit{WAGE}$ with the minimum number of flips. 

Let us discuss why this implementation is correct. Let $p$ be a peer that makes finitely many mistakes in a computation. In this computation, the value of $\mathit{flips}[p]$ is finite in all peers. If $p$ makes infinitely many mistakes, $\mathit{flips}[p]$ grows without a bound in all peers. By the specification of $\mathit{WAGE}$, for each block $b$, there is at least one process per partition that makes finitely many mistakes. Our implementation selects the output of implemented $\diamond AGE(b)$ to be the one with the smallest number of flips. This way, in any computation, eventually, the output of $\diamond AGE(b)$ is correct.
 \qed\end{proof}

Detector \emph{SPLIT} outputs whether a split in the system has occurred. Observe that the message delivery is unreliable only in case of a network split. 
Therefore, the occurrence of the split can be determined if \emph{PROP} indicates that a certain block has not propagated to the whole system.
The converse is not in general true. Just the fact of a split does not allow peers to determine whether the particular block has reached every peer. Hence the following lemma.

\begin{lemma}
The split detector is strictly weaker than the propagation detector. That is: $SPLIT \prec PROP$.
\end{lemma}

This theorem summarizes the above lemmas.

\begin{theorem}\label{thrmDetectors}
The relationship between partitioning detectors is as follows:
\[
PROP \equiv AGE \succ \diamond AGE \equiv \mathit{WAGE},\ PROP \succ SPLIT.
\]
\end{theorem}

\section{Algorithm \emph{PART}}
\label{section:part}

\begin{figure*}[t]
    \centering
    \includegraphics[width=14cm]{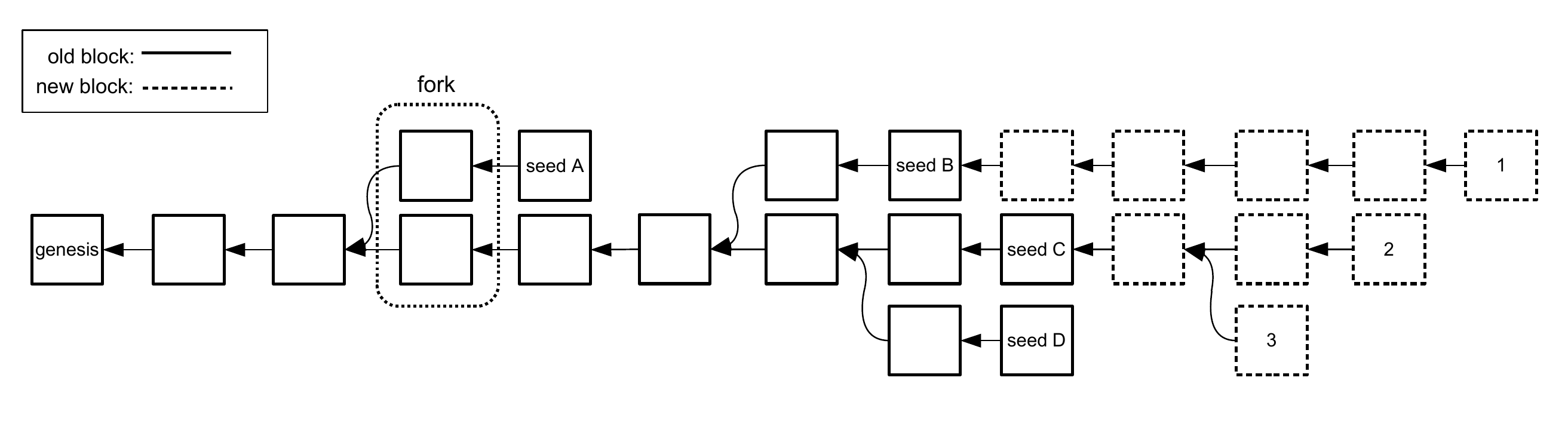}
    \vspace{-6mm}
    \caption{Global blockchain tree generated by \emph{PART}.}
    \label{figPARTImage}
\end{figure*}

In this section we present an algorithm, we call \emph{PART}, that solves the Partitionable Blockchain Consensus Problem. This algorithm is shown in Figure~\ref{figPART}.

\ \\
\textbf{Algorithm description.} The gist of the algorithm is as follows. The peers construct the blockchain tree and read the output of \emph{AGE} to classify which portion of this tree is old and which is new. This way the peers agree on the same old block to be the seed used for splitting the accounts. The peers then proceed to mine new blocks according to the split accounts in the separate partitions on the basis of this seed. 

Let us describe the algorithm in detail. Each peer maintains its copy of the blockchain $bc$,
and a priority queue $\mathit{txs}$ of all received transactions. The transactions are arranged in the order of their identifiers in $\mathit{txs}$. If a new transaction is received, it is entered into $\mathit{txs}$. We assume that $\mathit{txs}$ is never empty. See Figure~\ref{figPARTImage} for an example of a blockchain. 

\emph{PART} keeps track of the most recent output of the \emph{AGE} detector in the $currentAge$ variable. When $currentAge$ is $new$, the block is mined with split accounts. The value of $currentAge$ is recorded in the mined block.
Once the block is mined, \emph{PART} checks the output of \emph{AGE} against
the new block and sets $currentAge$ accordingly.

Function $mainChain()$ of \emph{PART} operates as follows. It consults \emph{AGE} for all blocks in $bc$ and constructs $trueTree$ with only those blocks whose recorded age matches the output of \emph{AGE}. If \emph{AGE} is perfect, the $trueTree$ contains all of $bc$. Function $mainChain()$ then builds $oldTree$ that contains only old blocks that are connected to the genesis block. Function $mainChain()$ then finds the longest branch $oldBrach$ in $oldTree$.  Ties are broken deterministically. Then, $mainChain()$ examines the new branches of $trueTree$ connected to the tail of $oldBranch$ and selects the longest branch which it stores in $newBranch$. Function $mainChain()$ returns the concatenation of $oldBranch$ and $newBrach$,
To summarize, $mainChain()$ operates on the subtree whose age agrees with the output of \emph{AGE} and returns the longest old branch connected to the longest new branch. We refer to the output of $mainChain()$ as \emph{main chain}.

Let us discuss the operation of $mainChain()$ using the example in Figure \ref{figPARTImage}. The $oldTree$ there includes all blocks in the branches that run from the genesis to seeds $A$, $B$, $C$ and $D$. Branches that run to $C$ or $D$ are of equal length and longer than the others.  Assume the tie is broken in favor of $C$. Two new branches are attached to $C$. The branch that runs to block $2$ is selected since it is the longer one. Function $mainChain()$ returns the branch that runs from the genesis to block $2$. Note that even though the branch that runs to $1$ is longer overall, old blocks are considered first. Thus, branch to $1$ has a shorter old blocks branch. 

Function $nextValid()$ returns the fist valid transaction in $txs$ that is not in the main chain. Function $resetMining()$ checks whether the peer is currently mining the next valid transaction, and if not, it restarts mining.

\emph{PART} operates as follows. Each peer continuously attempts to mine the first valid transaction in $\mathit{txs}$ that is not in the main chain of $bc$.
If mined, the recorded age of the block is compared to the output of the \emph{AGE} detector. If they are the same, the newly mined block is added to $bc$ and broadcast. If not, new age is recorded in $currentAge$, and the block is discarded. Then, the mining of the next  transaction starts. 

If the peer receives a block $nb$ mined by another peer, it checks if this block is linked to any of the blocks in $bc$. If not then this block is added to $unlinked$, which is a set of such blocks. If $nb$ is linked to $bc$, the peer inserts this block into $bc$ then checks if any of blocks in $unlinked$ may now also be linked to $bc$.

\ \\
\textbf{Correctness proof.}
\begin{lemma}\label{lemGrows}
The main chain of every peer increases indefinitely.
\end{lemma}
Intuitively, the peer either mines the blocks itself, or receives infinitely many mined blocks. Hence, the main chain keeps increasing. 

\begin{proof}
Since the communication is reliable, rather than consider separate trees maintained by each peer, we can consider a global tree of all blocks mined so far. The peers can be represented as located on this global main chain according to the mined blocks they have received. 

Let us consider the position of each peer on this global blockchain tree according to the peer's main chain. The peer may change its position by receiving a mined block or by mining a block itself. Let us discuss the position change due to receiving a block. If the peer receives a block linked to its current branch, it moves up the branch and extends its main chain. If the peer receives a block linked to a different branch, it may move to the new branch. Let us examine this move. Let $cb$ be the current branch and $nb$ be the new branch. 

Each branch is a concatenation of two chains: a prefix of old blocks and a suffix of new blocks. 
The move happens if the prefix of $nb$ is longer than the prefix of $cb$; or the prefixes are the same and the suffix of $nb$ is longer than the suffix of $cb$. 

If the computation contains no split, then every block that is mined is an old block. In this case, the new suffixes do not exist and the peer moves only if the prefix of $nb$ is longer. 

If the computation contains a split, the peer may potentially move to a shorter new branch because it has a longer prefix. 
However, if there is a split, the number of old blocks is finite. Once every peer receives all the old blocks, such moves are no longer possible. That is, the number of times the peer switches to a shorter branch is finite. After these switches, the peer may switch branches only if it has a longer suffix. That is, if this new branch is longer. Thus, if the peer switches position due to block receipt infinitely many times, its main chain grows indefinitely.

Let us consider a peer $p$, that changes position due to block receipt only finitely many times. This means that $p$ mines infinitely many transactions itself. In which case its main chain grows indefinitely as well. 
\qed\end{proof}

\begin{figure*}
\begin{tabular}{l@{\hskip 1in}c@{\hskip 0.5in}c}

\begin{minipage}[t]{0.5\textwidth}
\scriptsize
\label{part}
\begin{tabbing}
123\=123\=123\=12345\=12345\=12345\=12345\=12345\=12345\=12345\=\kill
$\textbf{variables}$\\
\>  $bc$ \ \ \ // tree of mined blocks, rooted in genesis \\
\>  $unlinked$ \ \ \ // set of received blocks with missing intermediate \\ 
\> \> \> // links \\
\>  $\mathit{txs}$ \ \ \ // queue of received transactions, prioritized by id \\
\> $currentAge$ \ // true if accounts are split after partitioning \\\
\ \\
$\textbf{functions}$ \\
\>$mainChain()$ \\
\>\> $trueTree$ := blocks in $bc$ whose age match $AGE$ output \\
\>\> $oldTree$ := branches with only old blocks of \\
\>\>\>\>\> $trueTree$ rooted in genesis\\ 
\>\> $oldBranch$ := longest branch in $oldTree$ \\ 
\>\> $newTree$ := branches with only new blocks of\\
\>\>\>\>\> $trueTree$ rooted in $tail(oldBranch)$\\
\>\> $newBranch$ := longest branch in newTree \\ 
\>\> $\textbf{return}\ oldBranch + newBranch$ \\ 
\ \\
\>$nextValid()$ // returns the first valid \\
\>\>\>\>\ \ \  // transaction in $\mathit{txs}$ that is not in $main(bc)$\\
\ \\
\> $resetMining()$ \\
\>\> \textbf{if} not mining $nextValid()$ \\
\>\>\> startMining $nextValid()$ with $currentAge$\\
\end{tabbing}
\end{minipage}
&
\begin{minipage}[t]{0.5\textwidth}
\scriptsize
\begin{tabbing}
123\=123\=123\=12345\=12345\=12345\=12345\=12345\=12345\=12345\=\kill
$\textbf{commands}$ \\

\textbf{receive} transaction $t$ $\longrightarrow$ \\
\> $enqueue(txs, t)$\\
\> $resetMining()$ \\
\ \\ 

\textbf{mine} block $nb$ $\longrightarrow$ \\ 
\> $currentAge$ := $AGE(nb)$ \\
\> $insert(nb,bc)$\\
\> $broadcast(nb)$ \\
\> $resetMining()$ \\

\ \\
\textbf{receive} mined block $nb$ $\longrightarrow$ \\ 
\>   \textbf{if} possible to $insert$($nb$,$bc$) \ \ \ // add new block to blockchain\\
\>\>    $insert(nb,bc)$ \\
\>\> \textbf{while} exists $b$ in $unlinked$ that can be inserted into $bc$ \\
\>\>\>  $insert(b, bc)$ \\
\> \textbf{else} \\
\>\> add $nb$ to $unlinked$ \\
\> $resetMining()$ 
\end{tabbing}
\end{minipage}
\end{tabular}
\captionof{figure}{Algorithm \emph{PART}.}\label{figPART}

\begin{tabular}{c c}
\begin{minipage}[b]{0.5\textwidth}
\scriptsize
\begin{tabbing}
123\=123\=123\=12345\=12345\=12345\=12345\=12345\=12345\=12345\=\kill
$\textbf{constants} $ \\
\> $p$ // identifier of this peer\\
\> $b$ // block whose age is evaluated \\
\ \\
 $\textbf{variables}$\\
\>$\mathit{flips}$ // array of changes in output from each peer, initially zero \\
\>$ages$ //array of most recent outputs of \emph{WAGE}, initially old\\
\ \\
\textbf{commands}\\
$ages[p]\ \textbf{not}=\ WAGE(b) \longrightarrow$ \\
\>$ages[p] = WAGE(b)$ \\
\>increment $\mathit{flips}[p]$ \\
\>$broadcast(\mathit{flips}[p], ages[p])$\\
\ \\
\textbf{receive} $numFlips, age$ \textbf{from} $id$ $\longrightarrow$ \\
\> $\mathit{flips}[id] := numFlips $ \\
\> $ages[id] := age $ \\
\ \\
$\diamond AGE(b) \longrightarrow$ // implemented output of $\diamond AGE$ \\ 
\> $\textbf{output}\ ages[id]$ for the $id$ with minimum $\mathit{flips}[id]$ \\
\end{tabbing}
\caption{Implementation of $\diamond AGE$ using \textit{WAGE}.}\label{figAgeEquivalency}
\end{minipage}
&
\begin{minipage}[t]{0.49\textwidth}
\includegraphics[width=\textwidth]{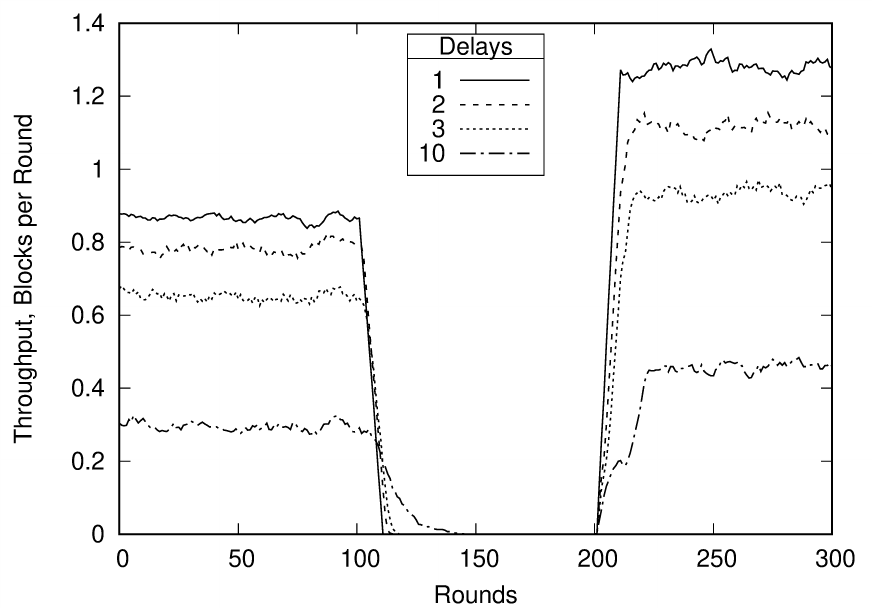}
\caption{\emph{PART}+$\diamond$\emph{AGE}. Split at round 100, detector recognizes it at round 200.}
\label{fig:LaggingWage}
\end{minipage}
\\
\begin{minipage}[t]{0.49\textwidth}
\includegraphics[width=\textwidth]{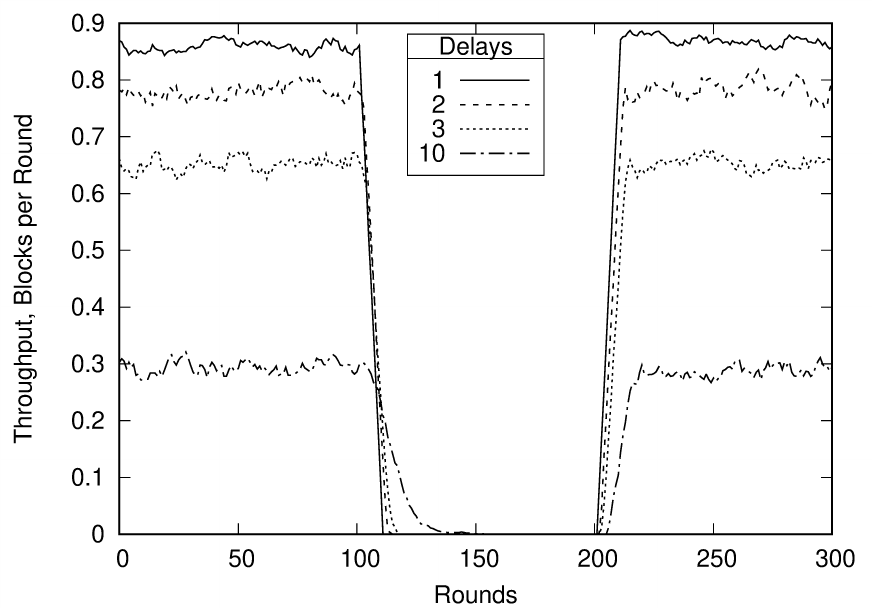}
\caption{\emph{PART}+$\diamond$\emph{AGE}. No split in computation, detector mistakenly recognizes it at round 100, corrects at round 200.}
\label{fig:LyingWage}
\end{minipage}
&
\begin{minipage}[t]{0.49\textwidth}
\includegraphics[width=\textwidth]{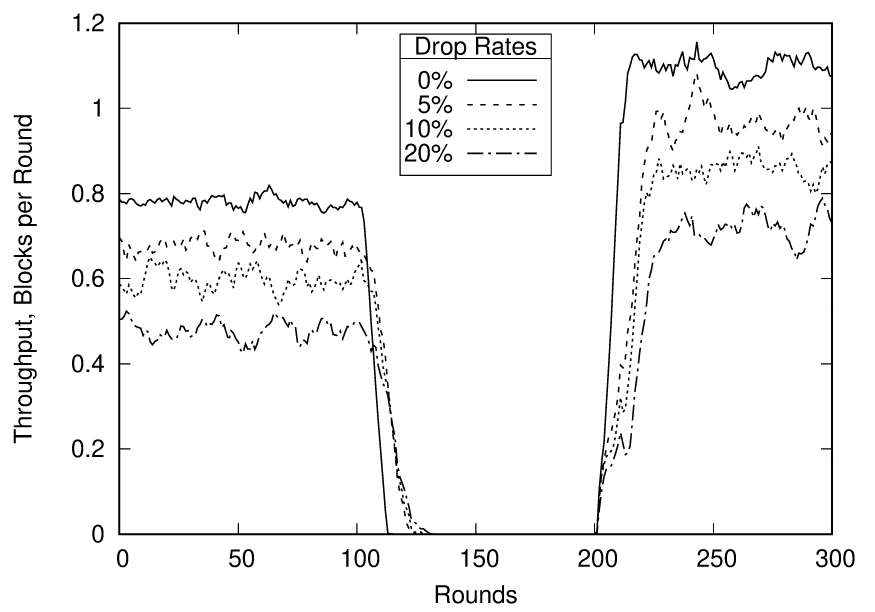}
\caption{\emph{PART} + $\diamond$\emph{AGE} with message loss. Split at round 100, detector recognizes it at round 200. Message delay is 2 rounds.}
\label{fig:MessageLoss}
\end{minipage}
\end{tabular}
\end{figure*}

\begin{lemma}\label{lemProgress}
 $PART$ satisfies the progress property of the Partitionable Blockchain Consensus Problem.
\end{lemma}

\begin{proof}
 Let $t$ be a permanently valid transaction. Let us consider the suffix of the computation where each peer $p$ receives $t$. Once $p$ receives $t$, it enters $t$ into $\mathit{txs}$. 
Each client submits transactions in the increasing order of their identifiers. That is, in any computation, there is only finitely many transactions with identifiers smaller than $t$. 

By Lemma~\ref{lemGrows}, the main chain of every peer increases indefinitely. Every block in this main chain is mined by some peer. Hence, there must be at least one peer $p$ that mines infinitely many blocks in this main chain. 
By the design of the algorithm, each peer mines the valid transaction with the smallest identifier that is not in its main chain. Therefore, $p$ may only mine finitely many blocks before $t$. Thus, an infinite main chain must contain $t$.
\qed\end{proof}

\begin{lemma}\label{lemmaPrefix}
In $PART$, the permanent branches of all peers have a common prefix up to and including the seed block.
\end{lemma}
Intuitively, due to reliable message transmission, blocks mined before the split will be received by all peers. 

\begin{proof}
Since the message propagation is reliable, blocks mined before split will be sent and received by all peers of the network. $PART$ computes the main chain to be the branch that is the longest distance from the genesis to the seed. The ties are deterministically broken. Therefore, once all peers receive all pre-split messages, their main chains include the same branch and the same seed. 
\qed\end{proof}


\begin{lemma}\label{lemBranch}
Algorithm $PART$ satisfies the branch compatibility property of the Partitionable Blockchain Consensus Problem. 
\end{lemma}

\begin{proof}
According to Lemma~\ref{lemmaPrefix}, the main chains of all peers include the same seed block. Due to Proposition~\ref{propMerge}, branches from different partitions are mergeable. Hence the lemma.
\qed\end{proof}

\begin{theorem}\label{thrmPART}
Algorithm $PART$ solves the Partitionable Blockchain Consensus Problem.
\end{theorem}
\begin{proof}
\emph{PART} never confirms an invalid transaction. Thus, \emph{PART} satisfies confirmation validity. Branch compatibility satisfaction is shown in Lemma~\ref{lemBranch}. Progress satisfaction by \emph{PART} is shown in Lemma~\ref{lemProgress}.
\qed\end{proof}

Observe that the presented algorithm operates correctly even if the detector makes finitely many mistakes. That is, if the detector is $\diamond AGE$. We call this combination $PART+\diamond AGE$. Indeed, once $\diamond AGE$ converges 
to the correct output for all blocks and each peer receives all pre-split blocks, all peers agree on the seed. From this point, \emph{PART} operates as with perfect \emph{AGE}. 

Per Theorem~\ref{thrmDetectors}, $\diamond AGE$ may be implemented using a weaker detector \emph{WAGE}. Let $PART+\mathit{WAGE}$ be the combination of \emph{PART} and such an implementation. Such a combination also solves the partitioning problem. Hence the following theorem.

\begin{theorem}\label{thrmPART+WAGE}
Algorithm $PART+\diamond AGE$ and $PART+\mathit{WAGE}$ solve the Partitioning Blockchain Consensus Problem.
\end{theorem}


\section{Performance Evaluation}
\label{section:evaluation}
\noindent\textbf{Setup.} We evaluate the performance of \emph{PART} using an
abstract simulation. We study the behavior of our algorithm through computations that we construct. The code for our simulation is available on GitHub~\cite{github}.

The simulated network consists of $n$ peers.
An individual computation is a sequence of rounds. In every round, each peer may receive new messages from each other peer, do local computation, and send messages to other peers. 

Each peer in the network has a unique channel to every other peer. Message delivery is FIFO. In a single round, a peer may receive messages from each sender.
Message propagation may take several rounds. Each message is delayed by a number of rounds. This delay is selected 
uniformly at random. 
That values range from $1$ to maximum delay $d$. Concurrent messages from the same sender do not impede other messages. That is, multiple messages from the same sender may be received in the same round.


The transaction submission rate is constant: one transaction per round. A submitted transaction is broadcast by a randomly selected peer. Block mining is simulated. Mining time is as follows. Each peer has an oracle that tells the peer whether it mined a block. In every round, the probability of mining a block for each peer is uniformly distributed between 1 and $d\cdot n$. The network size is $100$ peers. A split separates the network into two equally sized partitions of $50$  
peers. Overall transaction rate or mining rate does not change in the event of a split.

We measure algorithm throughput: the ratio of confirmed to submitted
transactions. We measure throughput under various delays. We plot the rolling average of number of confirmed transactions over the last $10$ rounds. We run $100$ experiments per each data point. 
To eliminate startup effects, we do not plot the first $100$ rounds. 



\ \\
\noindent\textbf{Results and analysis.} The results of our experiments are shown in Figures~\ref{fig:LaggingWage}, \ref{fig:LyingWage}, \ref{fig:MessageLoss}, and~\ref{fig:MultipleSplits}.

In Figure~\ref{fig:LaggingWage}, we show the results of the experiments with \emph{PART} and $\diamond AGE$.
The split occurs in round $100$. However, $\diamond AGE$ continues to classify all blocks as old until round $200$. The detector operates correctly afterwards. Between rounds $100$ and $200$, while the detector classifies blocks incorrectly, transactions are not confirmed. Therefore, the throughput decreases. Once the detector corrects itself, the old blocks are ignored by the algorithm, new blocks are mined and the throughput recovers,
This post-split throughput of the algorithm is higher since, instead of causing forks, blocks are concurrently confirmed in the two partitions. 

Figure~\ref{fig:LyingWage} also shows the results of \emph{PART} with $\diamond AGE$. In this case, there is no split, however, between rounds $100$ and $200$, the detector classifies all blocks as new. The detector recovers and starts classifying all blocks as old after round $200$. The algorithms behavior is similar to that shown in Figure~\ref{fig:LaggingWage}. Note, that there is no actual split in this experiment. Therefore, after the detector recovers, the number of forks does not decrease. 


\section{Mulitple Splits, Message Loss, Partition Merging}


\noindent
\textbf{Multiple splits.}
Let us consider the case of multiple split actions in a single computation. That is, a network partition may further split. Each individual split event separates a partition into two. We introduce two new detectors to handle this case. 

The perfect multiple block age detector \emph{MAGE} is the modification of \emph{AGE}. \emph{MAGE} operates as follows. For each block $b$, \emph{MAGE} outputs the number of split events that happened in the partition where this block is mined. A way to think about \emph{MAGE} is to consider that it repeatedly acts as a single-split \emph{AGE} detector for each partition. 
For example, if the partition is never split, \emph{MAGE} outputs $0$. If the network splits into two partitions $A$ and $A'$, \emph{MAGE} outputs $1$ for the peers of both partitions. If $A$ splits into $B$ and $B'$, then \emph{MAGE} outputs $2$ for the peers in $B$ and $B'$ and still $1$ for the peers of $A'$. 

The eventual multiple block age detector $\diamond \mathit{MAGE}$ and weak multiple block age detector \emph{WMAGE} are defined similarly to their single-split counterparts. 

Algorithm \emph{PART} operates with \emph{MAGE}, $\diamond \mathit{MAGE}$ and \emph{WMAGE} without modifications. We summarize this in the following theorem.

\begin{theorem}
Algorithms $\mathit{PART}+\mathit{MAGE}$, $\mathit{PART}+\diamond \mathit{MAGE}$, and $\mathit{PART}+\mathit{WMAGE}$ solve the Partitionable Blockchain Consensus Problem with multiple splits.
\end{theorem}

\begin{figure}
\includegraphics[width=\columnwidth]{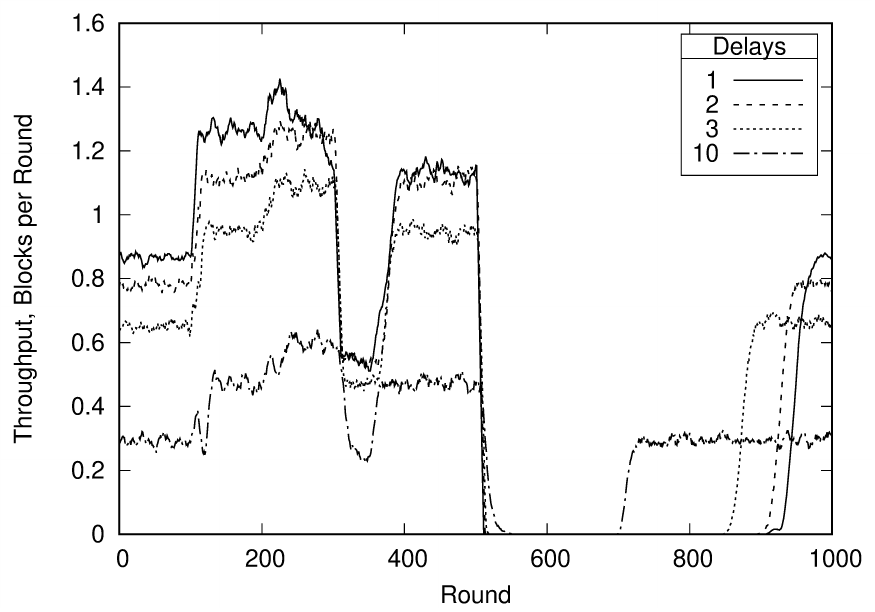}
\caption{\emph{PART} with perfect \emph{AGE} and multiple splits. No message loss. Split at round 100, second split at round 200. Merge second split at round 300, merge first split at round 500.}
\label{fig:MultipleSplits}
\end{figure}

\ \\
\noindent
\textbf{Message loss.} To tolerate message loss, the peers in algorithm \emph{PART} need to be able to recover lost messages. \emph{Disconnected subtree} is a collection of linked blocks not connected to the genesis block. Such collection has a single root. If a peer has a disconnected subtree, it is missing a block linking it to the genesis. This block may be delayed or lost.

The \emph{block catchup procedure} recovers the missing blocks. It contains two actions: (i) if the peer contains a disconnected subtree, broadcast request for the block preceding its root; (ii) if a peer receives such request and has the requested block, broadcast this block. Note that the correctness of the block catchup procedure does not depend on the timing of the request action so long as it is eventually executed. 

To be able to guarantee meaningful liveness, we restrict message loss as follows: if the same message is broadcast by the same peer infinitely many times, it is also delivered infinitely many times. If intermediate blocks are not delivered to a partition, there is no way to recover it after the split. Hence, we place another assumption. If a block $b$ is delivered to one of the peers in the partition, every block in the branch of $b$, i,e, on the chain from the genesis to $b$ is also delivered to one of the peers in this partition. With these assumptions we are able to state the following theorem.

\begin{theorem} Algorithm \emph{PART} with block catchup procedure solves the Partitionable Blockchain Consensus Problem with message loss.
\end{theorem}

Figure~\ref{fig:MessageLoss} evaluates the operation of \emph{PART} with message loss. 

\ \\
\noindent
\textbf{Partition merging.}
There are two ways of handling temporary split. It can be considered a special case of message loss. In this case the above message loss version of \emph{PART} operates correctly. However, in this \emph{competitive merge}, the blocks mined by one of the partitions are discarded. This may be inefficient. Alternatively, we may implement  \emph{cooperative merge} that retains some of the blocks of both partitions. In this case \emph{PART} and detectors need to be modified.

Detector \emph{SMAGE} for split-merge \emph{AGE}, correctly identifies whether the block was mined when partition was split. Detectors $\diamond \mathit{SMAGE}$ and \emph{WSMAGE} are defined similarly. To implement cooperative merge, algorithm \emph{SMPART} modifies the original \emph{PART} as follows. All peers run block catchup procedure that counters message loss. By examining the output of the detector on the received blocks, each peer determines that there was a split and then a merge. Each peer then finds the leaves of the two 
longest non-conflicting branches formed during partition. The peer then mines \emph{merge block} that, rather than to a single block, links to these two blocks. This way, both branches are confirmed. If multiple such merge-blocks are mined, the block on the longest branch wins. Ties are broken deterministically.

After the merge, the merge-block acts as a seed, the accounts are combined and the computation proceeds as pre-split. The algorithm and detectors can be extended to multiple partition merge similar to multiple partition split. In case of multiple splits and merges, the accounts are combined on the basis of the account share each partition had pre-merge. Consider an example. Assume $50/50$ split. If the network splits into two $A$ and $A'$, then $A$ splits into $B$ and $B'$, and then $B'$ merges with $A'$ into $C$. Then, $B$ contains $25\%$ of the account balances while $C$ contains $75\%$.

Figure~\ref{fig:MultipleSplits} shows the results of the experiments with thus implemented algorithm \emph{SMPART} and \emph{SMAGE} detector. The detector does not make mistakes. There is no message loss. There are two consecutive splits that then merge back. 

\section{Other Extensions and Future Work}
\noindent
\textbf{Detector implementation, block purging.} The pure asynchronous system allows us to reason about the essential properties of the algorithm that do not rely on timing assumptions. Nonetheless, we would like to outline certain implementation and usage aspects of the proposed algorithm and detectors.

The age detector may be implemented with checkpointing. The idea is as follows. The peers agree on a checkpoint block on every branch of the blockchain. Once the split occurs, if a certain block precedes the checkpoint block, it is considered old. A block is new if it follows the checkpoint block. To limit the rollback overhead, the checkpoints are moved closer to the leaves of the blockchain as the computation progresses.

Checkpoints can also be utilized to save memory space used to store the blockchain branches. Since the peers never roll back 
past checkpoint blocks, rather than storing individual old blocks, it is sufficient to just store the resultant checkpoint account balances. The old blocks may then be purged from memory.

\ \\
\noindent\textbf{Fault tolerance.} Let us discuss how the proposed partitionable blockchain may withstand other faults. Peer crashes may be problematic as the blockchain has no way of determining whether the split partition is operational or it crashed. In this case, the crashed partition leads to the loss of its share of account balances. To enable crash tolerance, a crash failure detector~\cite{chandra1996unreliable,chandra1996weakest} may need to be incorporated in the design.

A robust blockchain needs to be tolerant to Byzantine faults~\cite{lamport2019byzantine} where affected peers may behave arbitrarily. Byzantine peers may compromise agreement on the seed block or on the split itself. We believe that our proposed algorithm may be made tolerant to such faults. However, a definitive study is left for future research.

\bibliographystyle{IEEEtran}
\bibliography{references}

\end{document}